\documentclass[9pt,technote]{IEEEtran}
\usepackage{pbalance}
\usepackage{graphicx}  
\usepackage{epstopdf}   
\usepackage{amsmath,amssymb,amsfonts}
\usepackage{algorithmic}
\usepackage{algorithm}
\usepackage{array}
\usepackage[caption=false,font=normalsize,labelfont=sf,textfont=sf]{subfig}
\usepackage{textcomp}
\usepackage{stfloats}
\usepackage{url}
\usepackage{verbatim}
\usepackage{cite}
\usepackage{subfig}
\usepackage[thmmarks,amsmath]{ntheorem}
\theoremseparator{:}

\newtheorem{lemma}{Lemma}
\newtheorem{proposition}{Proposition}
\theoremheaderfont{\it}
\theorembodyfont{\normalfont}
\theoremsymbol{\ensuremath{\Box}}
\newtheorem*{proof}{Proof}
\theoremheaderfont{\it}
\theorembodyfont{\normalfont}
\theoremsymbol{}
\newtheorem*{remark}{Remark}
\begin{document}
	
	\title{Mission Efficiency Optimization in Low-Altitude Economy: Adaptive Power Allocation for Coordinating Heterogeneous Aircraft Swarms}

	\author{Jiarui Zhang,~\IEEEmembership{Student Member,~IEEE}, Wei Feng,~\IEEEmembership{Senior Member,~IEEE}, Chao Dong,~\IEEEmembership{Senior Member,~IEEE}, Ning~Ge,~\IEEEmembership{Member,~IEEE}, and Qihui Wu,~\IEEEmembership{Fellow,~IEEE}
		
		\thanks{J. Zhang, W. Feng, and N. Ge are with the Department of Electronic Engineering, State Key Laboratory of Space Network and Communications, Tsinghua University, Beijing 100084, China (e-mail: zhangjr22@mails.tsinghua.edu.cn, fengwei@tsinghua.edu.cn, gening@tsinghua.edu.cn).}	
		\thanks{C. Dong and Q. Wu are with the Key
			Laboratory of Dynamic Cognitive System of Electromagnetic Spectrum Space,
			Nanjing University of Aeronautics and Astronautics, Nanjing 211106, China (e-mail: dch@nuaa.edu.cn, wuqihui@nuaa.edu.cn).}
	}

	\maketitle
	
	\begin{abstract}
		With the rapid development of the low-altitude economy, low-altitude operations are booming, where complex missions require collaborative efforts among multiple heterogeneous low-altitude aircrafts (LAAs). Specifically, different LAAs assume distinct roles: some for sensing, some for communication, some for computing, and others for mission execution, together forming a sensing-communication-computing-control ($\textbf{SC}^3$) closed loop, akin to a \emph{reflex arc}. To enable efficient coordination in such multi-LAA swarms, we introduce the concept of operational-capability entropy (OCE) to quantify the effective work capability of operational LAAs. Accordingly, by jointly considering heterogeneous OCE and channel conditions among LAAs, we formulate the power allocation problem with the goal of minimizing the linear quadratic regulator (LQR) cost, which serves as a metric for mission efficiency. The resulting complex optimization problem is decomposed into two convex subproblems that are solved iteratively, with closed-form solutions derived for each. Simulation results demonstrate that the proposed mission-oriented adaptive power allocation scheme significantly outperforms traditional ones.
	\end{abstract}
	
	
	\section{Introduction}
	The low-altitude economy is growing rapidly, relying on diverse low-altitude aircrafts (LAAs)~\cite{LAE1,LAE2,LAE3}. Beyond logistics, inspection, and tourism, LAAs are increasingly deployed in low-altitude operations to enable rapid response and replace humans in hazardous or inaccessible environments~\cite{LAE1,LAE2}. Consider high-rise fire scenarios, where ground equipment often fails to reach the fire source and urgent rescue operations must be conducted at low altitudes by LAAs. Given the limited capability of a single LAA, multi-LAA coordination becomes essential, in which multiple LAAs equipped with sensing, communication, and computing payloads collaborate to accomplish complex missions. For example, in the high-rise fire scenario, sensing LAAs equipped with thermal imagers and cameras collect environmental data and upload it to computing LAAs, which process the data to generate control commands and transmit them to operational LAAs for firefighting execution. This entire process forms a sensing-communication-computing-control ($\textbf{SC}^3$) closed loop, akin to a \emph{reflex arc} \cite{EIH1,EIH2,Lei,Fang}. To collect environmental information and coordinate the swarm, the computing LAAs are also equipped with communication modules to bridge the sensing LAAs and the operational LAAs, thereby serving as edge information hubs (EIHs)~\cite{EIH1,EIH2,Lei}. Acting as a novel network element in unmanned operations, the EIH functions much like the nerve center in a \emph{reflex arc}.
	
	\IEEEpubidadjcol
	
	In low-altitude operations, whether control commands can be successfully delivered and executed directly determines the overall performance and mission efficiency~\cite{control,system}. Therefore, the downlink wireless links that deliver control commands from the EIH to multiple operational LAAs are critical to the entire $\textbf{SC}^3$ closed loop. However, the communication resources of the EIH are constrained by its limited payload~\cite{EIH1,EIH2}. Moreover, the operational LAAs exhibit significant heterogeneity in both their operational capabilities and channel conditions. Coupled with the strict power budget, such heterogeneity poses great challenges for efficient resource allocation across the swarm.
	
	Existing studies on downlink resource allocation have mainly focused on communication performance. For instance, the work in\cite{Hu} proposed a deep reinforcement learning (DRL)-based algorithm for multi-cell power allocation to maximize the sum rate. However, such communication-oriented designs focus solely on channel-level metrics, while neglecting the tight coupling among the sensing, communication, computing, and control stages of a mission. In the field of $\textbf{SC}^3$ closed loop control, mission-oriented resource allocation has recently attracted attention. Lei \textit{et al.} designed a control-oriented power allocation scheme for multiple $\textbf{SC}^3$ loops in integrated satellite-UAV networks\cite{Lei}, while Fang \textit{et al.} achieved intra-loop and inter-loop task balance through the joint optimization of bandwidth, time, and CPU frequency\cite{Fang}. Nevertheless, these studies model multi-actuator missions as multiple independent closed loops and thus fail to capture the collaborative nature of coordinated unmanned operations. Moreover, they implicitly assume that actuators possess unlimited operational capabilities. In practical low-altitude operations, however, each LAA has limited onboard processing, caching, and command decoding capabilities; commands delivered beyond these limits cannot be processed in time, which not only wastes communication resources but also introduces control delay and may even destabilize the closed loop system\cite{delay,stable}.

	Motivated by the above issues, we investigate multi-LAA coordinated operations from the perspective of $\textbf{SC}^3$ closed loop control. The main contributions are summarized as follows. First, we introduce the novel concept of operational-capability entropy (OCE) to quantify the effective work capability of an LAA during mission execution, which captures the physical limit on the amount of control information that each LAA can receive, decode, and utilize within one $\textbf{SC}^3$ cycle. Second, by jointly considering the heterogeneous OCE and channel conditions among LAAs, we formulate a mission-oriented adaptive power allocation problem that minimizes the linear quadratic regulator (LQR) cost, thereby directly optimizing mission efficiency rather than communication metrics alone. Third, the resulting non-convex problem is judiciously transformed into a max-min form and decoupled into two convex subproblems, for which closed-form solutions are derived and an iterative algorithm with guaranteed convergence is developed. Simulation results demonstrate that the proposed scheme effectively exploits the limited power budget by balancing OCE and channel heterogeneity, significantly outperforming traditional approaches and validating the critical role of OCE in closed loop control performance.
	
	\section{System Model and Problem Formulation}
	Fig.~\ref{fig:system} illustrates a multi-LAA cooperative system built upon the $\textbf{SC}^3$ closed loop control, where $K$ operational LAAs collaborate to accomplish the same mission. Assisted by the remote control center and cloud server via satellite backhaul, the EIH orchestrates the whole closed loop and, in particular, delivers control commands to the operational LAAs. In this letter, we focus on the downlink command-delivery links from the EIH to the $K$ operational LAAs over orthogonal channels. Denoting by $p_k$ the transmit power allocated to LAA $k$, the total power budget satisfies $\sum_{k = 1}^{K} p_k \leq P_{\text{max}}$, where $P_{\text{max}}$ is the maximum transmit power of the EIH.
	
	\begin{figure} [t]
		\centering
		\includegraphics[width=1.0\linewidth]{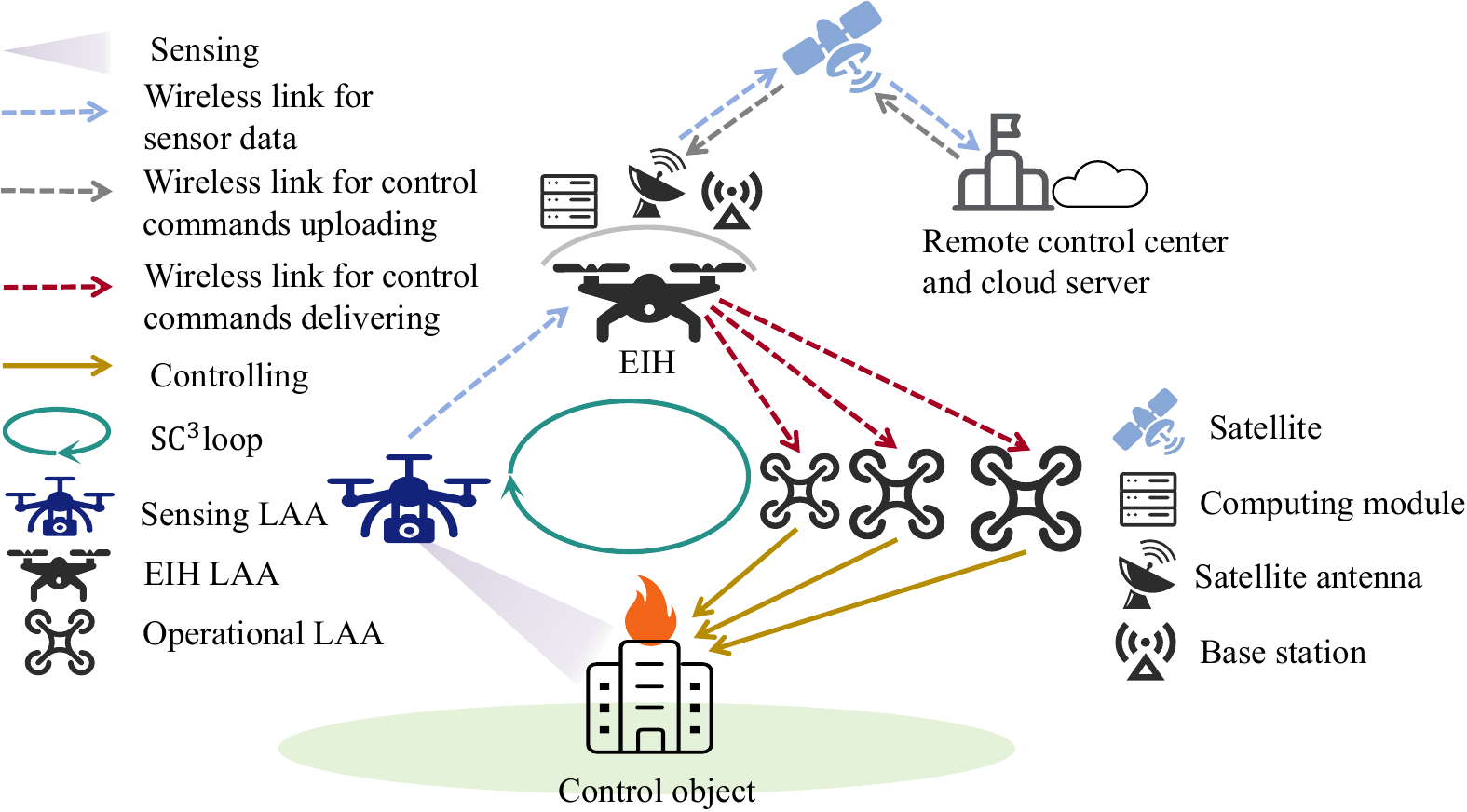}
		\caption{Illustration of a low-altitude emergency rescue scenario, where heterogeneous aircraft are dispatched for sensing, edge information hub, and coordinated unmanned operation. They together form a \emph{reflex arc}-like $\textbf{SC}^3$ closed loop system.}
		\label{fig:system}
	\end{figure}
	
	Assuming both the operational LAAs and the to-user communication module on the EIH are equipped with a single antenna, the ergodic capability of the EIH- $k$th LAA channel, denoted by $C_k$, is given by 
	\begin{equation}\label{xindao}
		C_k = \mathbb{E} \left[ \log_2 \left( 1 + \frac{p_k |h_k|^2}{\sigma^2} \right) \right],
	\end{equation}
	where $\mathbb{E}\left[\cdot\right]$ denotes mathematical expectation and $\sigma^2$ is the noise variance. Additionally, $h_k$, the channel gain between the EIH and LAA $k$, is modeled as
	\begin{equation}\label{zengyi}
		h_k=s_k \cdot l_k,
	\end{equation}
	where $s_k \overset{i.i.d}{\sim} \mathcal{CN}(0,1)$ denotes the fast-varying small-scale channel gain, i.e., Rayleigh fading. Meanwhile, $l_k$ represents the large-scale channel gain, and is expressed as
	\begin{equation}\label{dachidu}
		l_k = \sqrt{(d_k)^{-\gamma} F_k},
	\end{equation}
	where $d_k$ is the distance between the EIH and LAA $k$, $\gamma$ is the path loss exponent, and $F_k$ represents shadow fading. From the work in \cite{Feng}, we can get a closed-form approximation for $C_k$ as
	\begin{align}\label{jinsi}
		C_{k,ap} = \min_{w_k \geq 0} \left\{ \log_2 \left( 1 + \frac{p_k l_k^2}{e^{w_k} \sigma^2} \right) + (w_k + e^{-w_k} - 1) \log_2 e \right\}
	\end{align}
	where $w_k$ is the auxiliary variable. Let $R$ be the mission-related data rate of the whole $\textbf{SC}^3$ loop per cycle, and $R_k$ be the rate of LAA $k$. We have $R\leq\sum_{k = 1}^{K} R_k$ and $R_k\leq B_kTC_{k,ap}$, where $T$ is the available transmission time for control commands within one $\textbf{SC}^3$ loop, and $B_k$ is the channel bandwidth.
	
	\begin{figure} [t]
		\centering
		\includegraphics[width=1.0\linewidth]{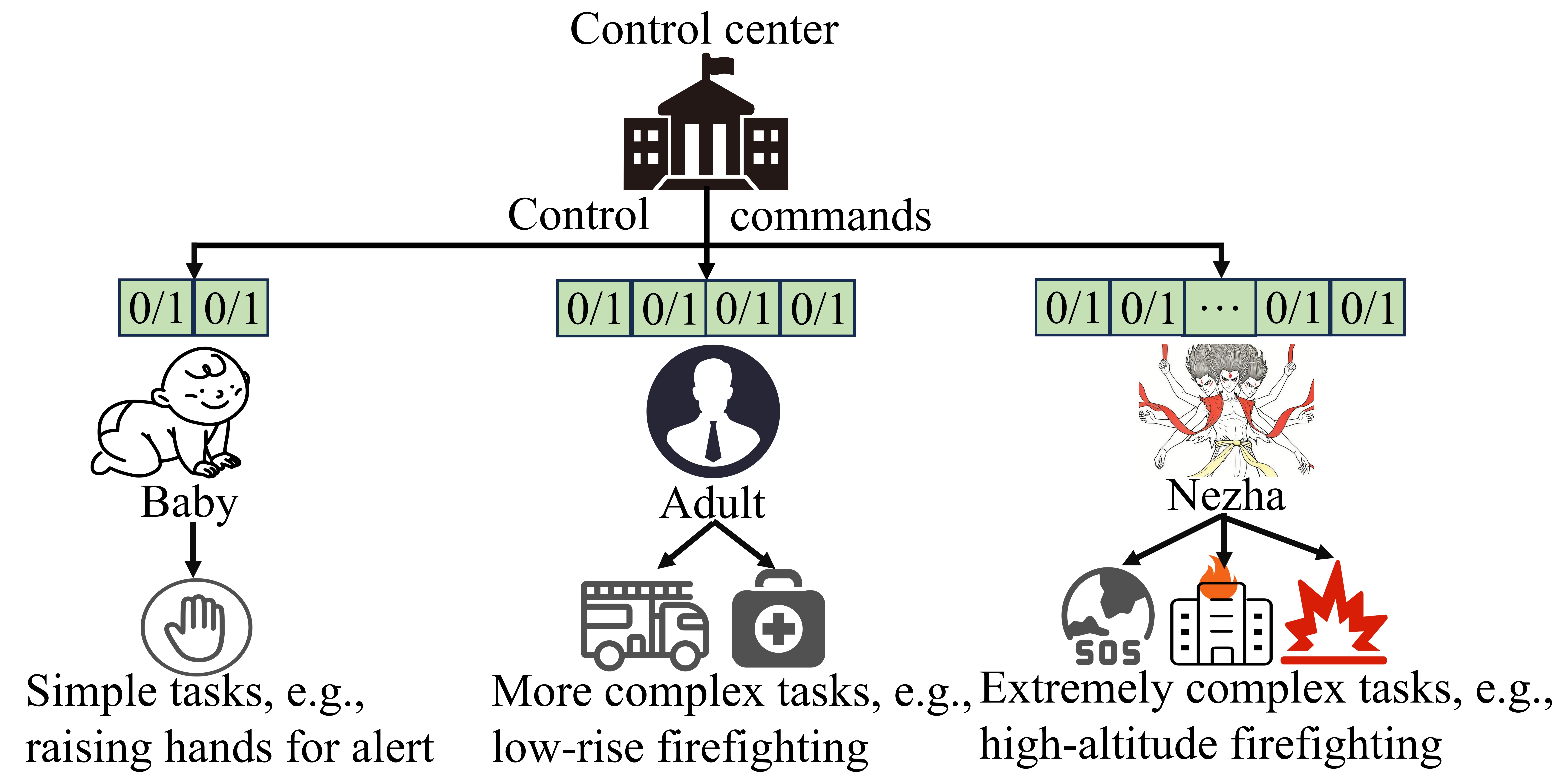}
		\caption{Illustration of OCE constraints. The control center issues commands to three agents with increasing capabilities, reflecting different receivable instruction count and executable mission complexity in firefighting scenario.}
		\label{fig:OCE}
	\end{figure}
	
	To quantify the effective work capability of operational LAAs, we propose the concept of operational-capability entropy (OCE). For the operational LAA $k$, its OCE $E_k$ is defined as the upper bound on the amount of control commands that can be effectively executed within one $\textbf{SC}^3$ cycle, measured in bits /$\textbf{SC}^3$ cycle. Here, ``effectively executed" means that the corresponding control commands can be received, decoded, and converted into actuation actions within the current control cycle. This upper bound arises from the inherent hardware limitations of operational LAAs. An LAA with a large OCE is assigned more control commands to undertake more numerous and complex tasks, whereas an LAA with a small OCE is assigned fewer commands to perform fewer and simpler tasks. Accordingly, the OCE constraint is given by $R_k \leq E_k$. Transmitted commands exceeding $E_k$ cannot be processed in time, resulting in inefficient use of communication and computing resources.
	
	Fig \ref{fig:OCE} gives an intuitive illustration of OCE. The control center issues control commands to three agents with distinctly different capabilities—a baby, an adult, and Nezha (a legendary Chinese mythological character, often depicted with three heads and six arms). These three agents differ significantly in both the number of instructions they can receive and the complexity of the missions they can execute. Even if numerous tasks are assigned to the baby, it cannot effectively execute them. This clearly shows that the amount of control information an agent can effectively utilize is bounded by its own OCE, rather than by the downlink wireless channel alone. So the OCE constraints cannot be neglected.
	
	For the control part of the $\textbf{SC}^3$ closed loop, the process of multiple operational LAAs executing a cooperative mission is modeled as a discrete linear time-invariant system\cite{control,system}:
	\begin{equation}\label{system}
		\mathbf{x}_{i+1} = \mathbf{A}\mathbf{x}_{i}+\mathbf{B}\mathbf{u}_{i}+\mathbf{v}_{i},
	\end{equation}
	where $i$ is the time index, $\mathbf{x}_{i} \!\in\! \mathbb{R}^{n\times 1}$ is the system state, $\mathbf{u}_{i} \!\in\! \mathbb{R}^{m\times 1}$ is the control action (the total action is the sum of individual operational LAA actions due to orthogonal sub-channels, $\mathbf{u}_{i}=\sum_{k=1}^{K} \mathbf{u}_{i,k} $), $\mathbf{v}_{i} \!\in\! \mathbb{R}^{n\times 1}$ is system noise. $\mathbf{A}\!\in\! \mathbb{R}^{n\times n}$ and $\mathbf{B}\!\in\! \mathbb{R}^{n\times m}$ are the state matrix and input matrix respectively.
	
	In control theory, the control performance can be measured by the LQR cost function~\cite{control}, which is a weighted summation of the state derivation and control input,
	\begin{equation}\label{LQR}
		l \triangleq \sup\lim\limits_{N\rightarrow \infty}\mathbb{E} \left[ \frac{1}{N} \sum_{i = 1}^{N} \left(\mathbf{x}_{i}^\text{T}\mathbf{Q}\mathbf{x}_{i} +\mathbf{u}_{i}^\text{T}\mathbf{R}\mathbf{u}_{i}\right) \right],
	\end{equation}
	where $\mathbf{Q}$ and $\mathbf{R}$ are semi-positive definite weight matrices. According to\cite{control}, the LQR cost has a lower bound as:
	\begin{equation}\label{xiajie}
		l \geq \frac{n N(\mathbf{v}) |\det \mathbf{M}|^{\frac{1}{n}}}{2^{\frac{2}{n}(R - \log_2 |\det \mathbf{A}|)} - 1} + \operatorname{tr}(\mathbf{\Sigma}_{\mathbf{v}} \mathbf{S}),
	\end{equation}
	where $N\!\left( \mathbf{v}\right) \!\triangleq\! \frac{1}{2\pi e}\exp \left( \frac{2}{n} h\!\left(  \mathbf{v} \right)  \right) $, $h\!\left(  \mathbf{v} \right)$ is the differential entropy of $\mathbf{v}$, $\mathbf{\Sigma}_{\mathbf{v}}$ is its covariance matrix, $\log_2 |\det \mathbf{A}|$ is the intrinsic entropy rate. $\mathbf{M}$ and $\mathbf{S}$ are solved by the following Riccati equations:
	\begin{equation}\label{fangcheng}
		\mathbf{S}=\mathbf{Q}+\mathbf{A}^\text{T}(\mathbf{S}-\mathbf{M})\mathbf{A},\mathbf{M}=\mathbf{S}^\text{T}\mathbf{B}(\mathbf{R}+\mathbf{B}^\text{T}\textbf{S}\mathbf{B})^{-1}\mathbf{B}^\text{T}\mathbf{S}.
	\end{equation}
	To ensure the system can be stabilized, the total mission-related data rate in each cycle needs to exceed the intrinsic entropy rate\cite{system},
	\begin{equation}\label{neibing}
		R>\log_2 |\det \mathbf{A}|,
	\end{equation}
	otherwise, the LQR cost is infinite.
	
	The formulated mission-oriented adaptive power allocation problem (P1) is as follows:
	\begin{subequations}\label{P1}
		\begin{align}
			\text{(P1)}\quad&\min_{\mathbf{p},\mathbf{w}}\quad l \label{P1a} \\ 
			\textit{s.t.}\quad &l \geq \frac{n N(\mathbf{v}) |\det \mathbf{M}|^{\frac{1}{n}}}{2^{\frac{2}{n}(R - \log_2 |\det \mathbf{A}|)} - 1} + \operatorname{tr}(\mathbf{\Sigma}_{\mathbf{v}} \mathbf{S}) \label{P1b}\\
			&R>\log_2 |\det \mathbf{A}| \label{P1c}\\
			&R\leq\sum_{k = 1}^{K} R_k \label{P1d}\\
			&\sum_{k = 1}^{K} p_k \leq P_{\text{max}} \label{P1e}\\
			&R_k\leq B_kTC_{k,ap}, \quad k = 1, 2, \cdots, K \label{P1f}\\
			&R_k\leq E_k, \quad k = 1, 2, \cdots, K \label{P1g}\\
			&p_k\geq 0, \quad  k = 1, 2, \cdots, K \label{P1h}\\
			&w_k\geq 0, \quad  k = 1, 2, \cdots, K \label{P1i}
		\end{align}
	\end{subequations}
	where $\mathbf{p}=\left[ p_1, p_2, \cdots, p_K \right]$ is power allocation vector and $\mathbf{w}=\left[ w_1, w_2, \cdots, w_K \right]$ is an auxiliary vector. Constraints (\ref{P1f}) and (\ref{P1g}) jointly incorporate Shannon capability and OCE. The problem (P1) is non-convex and needs to optimize two vectors, which requires further transformation for solution.
	
	\section{Mission-Oriented Adaptive Power Allocation}
	Given that the right-hand side of (\ref{P1b}) is a decreasing function of $R$, minimizing the LQR cost is equivalent to maximizing $R$. In addition, to maximize the mission-related data rate, each subchannel must achieve its Shannon capability, therefore, equalities are required in (\ref{P1d}) and (\ref{P1f}). According to (\ref{jinsi}), problem (P1) can be transformed into the following max-min problem (P2):
	
	\begin{subequations}\label{P2}
		\begin{align}
			\text{(P2)}\quad&\max_{\mathbf{p}}\min_{\mathbf{w}} \sum_{k=1}^{K}R_k(w_k,p_k) \label{P2a} \\
			\textit{s.t.}\quad& R_k(w_k,p_k)\leq E_k, \quad k = 1, 2, \cdots, K \label{P2b} \\
			&(\ref{P1e}),(\ref{P1h}),(\ref{P1i}), \notag
		\end{align}
	\end{subequations}
	where
	\begin{align}\label{Rk}
		R_k(w_k,p_k)\triangleq\frac{B_kT}{\ln2}\left\{ \ln \left( 1 + \frac{p_k l_k^2}{e^{w_k} \sigma^2} \right) +(w_k + e^{-w_k} - 1) \right\}
	\end{align}
	We omit the stability condition (\ref{P1c}) and test it after solving (P2). If the optimal $R^*$ satisfies the stability condition (\ref{neibing}), the LQR cost can be calculated by (\ref{xiajie}). Otherwise, the system cannot be stabilized and the LQR cost is infinite.

	\begin{lemma}\label{yinli}
		In the feasible region, the variable $w_k$ must be finite and cannot take infinite values.
	\end{lemma}
	\begin{proof}
		From (\ref{P1e}) and (\ref{P1h}), it is easy to know $p_k$ is finite. For a fixed $p_k$, assuming $w_k$ could be infinite, the left-hand side of (\ref{P2b}) would tend to infinity while the right-hand side remains finite, causing a contradiction. Thus, within the feasible region, $w_k$ must be finite.
	\end{proof}
	
	\begin{proposition}
		Problem (P2), which is strongly concave in $\mathbf{p}$ and strongly convex in  $\mathbf{w}$, can be decomposed into two convex optimization subproblems, (P3) and (P4), and solved using an iterative algorithm.
	\end{proposition}
	\begin{subequations}\label{P3}
		\begin{align}
			\text{(P3)}\quad&\max_{\mathbf{p}} \sum_{k=1}^{K}R_k(p_k) \label{P3a} \\
			\textit{s.t.}\quad &(\ref{P1e}),(\ref{P1h}),(\ref{P2b}). \notag
		\end{align}
	\end{subequations}
	\begin{subequations}\label{P4}
		\begin{align}
			\text{(P4)}\quad&\min_{\mathbf{w}} \sum_{k=1}^{K}R_k(w_k) \label{P4a} \\
			\textit{s.t.}\quad &(\ref{P1i}). \notag
		\end{align}
	\end{subequations}
	(P3) is a convex optimization subproblem in $\mathbf{p}$ with $\mathbf{w}$ fixed, where $R_k(p_k)\triangleq R_k(w_k^{\text{fixed}},p_k)$. Similarly, (P4) is a convex optimization subproblem in $\mathbf{w}$ with $\mathbf{p}$ fixed. In (P4), the OCE constraint (11b) is omitted. Once $\mathbf{p}$ is fixed, (P4) can be decoupled into $K$ independent subproblems, each minimizing $R_k$. Consequently, during the alternating solution process, if (P3) satisfies the OCE constraint, (P4) inherently satisfies it as well, so we can omit (11b) in (P4).
	\begin{proof}
		Given that the $K$ subchannels are orthogonal and independent, the Hessian matrix is diagonal and its diagonal entries share the same structure. Therefore, we investigate the properties of $R_k(w_k,p_k)$. By computing the second derivative of the function and applying \textbf{Lemma \ref{yinli}}, we can obtain:
		\begin{equation}\label{mp}
			\frac{\partial^2R_k}{\partial{p_k}^2}=\frac{-B_kTl_k^4 e^{-2w_k}\log_2e}{\sigma^4(1+p_k\frac{l_k^2}{\sigma^2}e^{-w_k})^2}<0
		\end{equation}
		\begin{equation}\label{mw}
			\frac{\partial^2R_k}{\partial{w_k}^2}=\frac{B_kTl_k^2p_ke^{w_k}}{\ln2\left(\sigma e^{w_k}+\frac{l_k^2p_k}{\sigma}\right)^2}+B_kTe^{-w_k}\log_2e>0
		\end{equation}
		\begin{equation}\label{lwp}
			\left\lvert\frac{\partial^2R_k}{\partial w_k\partial p_k}\right\rvert=\frac{B_kTl_k^2e^{w_k}}{\ln2\left(\sigma e^{w_k}+\frac{l_k^2p_k}{\sigma}\right)^2}<\infty
		\end{equation}
		Therefore, $R_k(w_k,p_k)$ is uniformly strongly concave-convex and can be decoupled into two convex optimization subproblems, (P3) and (P4), which can be solved using an iterative algorithm, such as the primal-dual steepest descent (PDSD) algorithm with global linear convergence\cite{PDSD}.
	\end{proof}
	
	\begin{proposition}
		The optimal solution to (P3) is as follows:
		\begin{equation}\label{zuiyoujie1}
			p_k^* =
			\begin{cases}
				p_k^{\text{upper}},\qquad\qquad\qquad{\text{if}}\quad\sum_{k=1}^{K}p_k^{\text{upper}}\leq P_{\text{max}},\\
				{\min \left\{ p_k^{\text{upper}},\max \left\{0,\frac{B_kT}{\lambda}-\frac{\sigma^2}{l_k^2}\right\}\right\},}{\text{otherwise.}}
			\end{cases}
		\end{equation}
		where 
		\begin{equation}\label{pupper}
			p_k^{\text{upper}}\triangleq\frac{\sigma^2e^{w_k}}{l_k^2}\left(e^{\frac{E_k\ln 2}{B_kT}-w_k-e^{-w_k}+1}-1\right),
		\end{equation}
		and $\lambda$ is the unique constraint satisfying $\sum_{k=1}^{K}p_k^*=P_{\text{max}}$.
	\end{proposition}
	\begin{proof}
		We can transform (\ref{P2b}) into $p_k\leq p_k^{\text{upper}}$ with $w_k$ fixed. When $\sum_{k=1}^{K}p_k^{\text{upper}}\leq P_{\text{max}}$ is satisfied, each subchannel is shown to reach the OCE limit, and the power allocation attains its upper bound $p_k^{\text{upper}}$, such that increasing power yields no further performance improvement. The optimal solution is then given by $p_k^*=p_k^{\text{upper}}$.
		
		When $\sum_{k=1}^{K}p_k^{\text{upper}}>P_{\text{max}}$, indicating insufficient total power, it is easy to find the optimal solution necessitates the full utilization of the transmit power, so $\sum_{k=1}^{K}p_k=P_{\text{max}}$. This standard convex optimization problem can be solved by constructing the Lagrangian dual problem and applying the Karush-Kuhn-Tucker (KKT) conditions\cite{convex}, yielding (\ref{zuiyoujie1}).
	\end{proof}
	
	\begin{proposition}
		The optimal solution to (P4) is as follows:
		\begin{equation}\label{zuiyoujie2}
			w_k^*=\ln \left(\frac{1+\sqrt{1+\frac{4p_kl_k^2}{\sigma^2}}}{2}\right).
		\end{equation}
	\end{proposition}
	\begin{proof}
		Since the variables $w_k$ are mutually independent, (P4) can be decoupled into $K$ single-variable convex optimization subproblems. Neglecting constraint (\ref{P2b}) and solving Lagrangian dual problem with the KKT conditions yields solution (\ref{zuiyoujie2}). Substituting (\ref{zuiyoujie2}) into (\ref{P2b}): if satisfied, (\ref{zuiyoujie2}) solves (P4); otherwise, the choice of $p_k$ renders (P4) infeasible. When solved iteratively, the iterations remain within the feasible region, guaranteeing that (P4) has a solution.
	\end{proof}
	
	We introduce the following notation. For the fixed $\mathbf{w}$, let $F(\mathbf{w})$ denote the optimal solution $\mathbf{p}^*$ of (P3) given by (\ref{zuiyoujie1}) and $f(\mathbf{w})$ the corresponding optimal value. For the fixed $\mathbf{p}$, let $G(\mathbf{p})$ denote the optimal solution $\mathbf{w}^*$ of (P4) given by (\ref{zuiyoujie2}) and $g(\mathbf{p})$ its optimal value. We adopt fixed step lengths for the iteration\cite{PDSD}. Let $R(\mathbf{w},\mathbf{p})=\sum_{k=1}^{K}R_k(w_k,p_k)$. We then set:
	\begin{equation}\label{steplength1}
		\lambda_{\mathbf{w},\mathbf{p}}=\inf \{ \text{smallest eigenvalue of } \nabla_{\mathbf{w}\mathbf{w}}^2 R(\mathbf{w},\mathbf{p}) \}
	\end{equation}
	\begin{equation}\label{steplength2}
		\lambda_{\mathbf{p},\mathbf{w}}=\inf \{ \text{smallest eigenvalue of } -\nabla_{\mathbf{p}\mathbf{p}}^2 R(\mathbf{w},\mathbf{p}) \}
	\end{equation}
	\begin{equation}\label{steplength3}
		M_{\mathbf{w},\mathbf{p}} = \sup \left\{ \left\| \nabla_{\mathbf{w}\mathbf{p}}^2 R(\mathbf{w}, \mathbf{p}) \right\| \right\}
	\end{equation}
	\begin{equation}\label{steplength4}
		\tilde{\sigma} = \frac{M_{\mathbf{w},\mathbf{p}}^2}{\lambda_{\mathbf{w},\mathbf{p}} \lambda_{\mathbf{p},\mathbf{w}}},
	\end{equation}
	 where $\mathbf{w}$ and $\mathbf{p}$ lie in the feasible region. We summarize the proposed algorithm in \textbf{Algorithm \ref{suanfa}}.
	
	\begin{algorithm}[t]
		\caption{Iterative Algorithm for Power Allocation}
		\label{suanfa}
		\textbf{Input:} $\textbf{SC}^3$ loop-related parameters: $E_k$, $T$ and $K$;\\
		\hspace*{1.5em} Control-related parameters: $n$, $m$, $\mathbf{A}$, $\mathbf{B}$, $\mathbf{Q}$, $\mathbf{R}$ and $\mathbf{\Sigma}_{\mathbf{v}}$; \\
		\hspace*{1.5em} Communication-related parameters: $l_k$, $\sigma^2$, $B_k$ and $P_{\text{max}}$; \\
		\hspace*{1.5em} The iteration terminating threshold $\delta$;\\
		\vspace{-1.5em}
		\begin{algorithmic}[1]
			\STATE Calculate $\mathbf{S}$ and $\mathbf{M}$ according to (\ref{fangcheng});
			\STATE Calculate $\tilde{\sigma}$ according to (\ref{steplength1})-(\ref{steplength4});
			\STATE \textbf{Initialization:}
			\STATE set iteration counter $s=0$, step lengths $\alpha = \beta = \min \left\{ 1, \frac{1}{2\tilde{\sigma}} \right\}$, $w_k^0=0$, and $p_k^0=F(w_k^0)$;
			\REPEAT
			\STATE Generate intermediate points:\\
			$\hat{\mathbf{w}}^{s+1}= (1-\alpha)\mathbf{w}^{s} + \alpha G(F(\mathbf{w}^{s}))$, \\
			$\hat{\mathbf{p}}^{s+1}= (1-\beta)\mathbf{p}^{s} + \beta F(G(\mathbf{p}^{s}))$;
			\STATE Update with backward feedback:\\ $\mathbf{w}^{s+1}=\text{argmin}\left\{f(\mathbf{w})|\mathbf{w}=\hat{\mathbf{w}}^{s+1}\ \text{or}\ \mathbf{w}=G(\mathbf{p}^{s})\right\}$,\\
			$\mathbf{p}^{s+1}=\text{argmax}\left\{g(\mathbf{p})|\mathbf{p}=\hat{\mathbf{p}}^{s+1}\ \text{or}\ \mathbf{p}=F(\mathbf{w}^{s})\right\}$;
			\STATE $s = s + 1$;
			\UNTIL \\
			{$\displaystyle\left\lvert \min\{f(\mathbf{w}^{s}),f(G(\mathbf{p}^{s}))\}-\max\{g(\mathbf{p}^{s}),g(F(\mathbf{w}^{s}))\} \right\rvert \leq \delta$}
			\STATE Calculate the optimal $R_k^*$ according to (\ref{Rk}) and calculate the optimal total data rate $R^*=\sum_{k=1}^{K}R_k^*$;
			\STATE Judge the stability condition (\ref{neibing}) and calculate the limit LQR cost, $l^*$, according to (\ref{xiajie});
		\end{algorithmic}
		\textbf{Output:} The optimal LQR cost, mission-related data rate, allocated power, auxiliary variables: $l^*$, $R^*$, $\mathbf{p}^*$, $\mathbf{w}^*$.
	\end{algorithm}
	\begin{remark}
		The initial point must be chosen within the feasible region, as guaranteed by the proposed algorithm; otherwise, the iteration may not proceed.
	\end{remark}
	
	\section{Simulation Results and Discussion}
	In this section, we present simulation results to evaluate the performance of the proposed scheme. Unless otherwise specified, the following simulation parameters are used. The number of operational LAAs is $K=5$, which are randomly distributed within a circular area with a radius of $5$km and an altitude of $1$km. The EIH is located at the center of the circle on the ground. Referring to \cite{para}, in high-rise fire scenarios, smoke and high temperature significantly increase wireless signal attenuation. Accordingly, we set $B_k=5$kHz, $F_k=8$dB, for $k=1,2,\cdots,K$, $\gamma=4$, $\sigma^2=-110$dBm. $T$ is set to $49.8$ms for all LAAs, considering the satellite latency\cite{Lei}. $E_k$ is randomly generated from the interval $[0,1500]$(bits/cycle). Control-related parameters are given by $n=m=1000$, $\log_2 |\det \mathbf{A}|=2000$, $\mathbf{R}=\mathbf{0}_{1000}$, $\mathbf{Q}=\mathbf{I}_{1000}$. The control noise is assumed to be independent Gaussian random variables with mean zero and variance $0.01$. The iteration terminating threshold is $\delta=10^{-6}$. 
	
	\begin{figure} [t]
		\centering
		\includegraphics[width=0.95\linewidth]{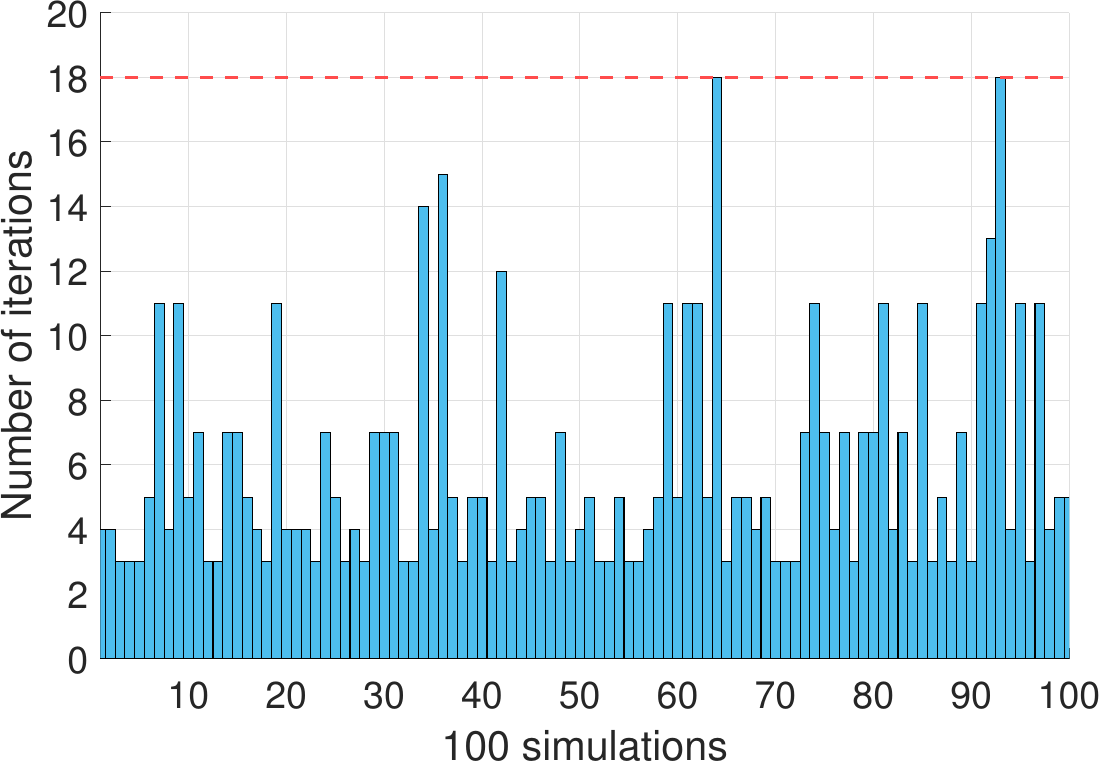}
		\caption{The number of iterations of the proposed iterative algorithm under $100$ random simulations.}
		\label{fig:sim1}
	\end{figure}
	
	Fig. \ref{fig:sim1} demonstrates the convergence performance of the proposed iterative algorithm. In the simulation, the number of operational LAAs, and the maximum transmit power are randomly selected from $[5,20]$ and $[6,24]$(dBW), respectively. The results show that the proposed algorithm converges within $18$ iterations across $100$ independent simulations, validating its low complexity and ability for quick convergence.
	
	\begin{figure} [t]
		\centering
		\includegraphics[width=0.98\linewidth]{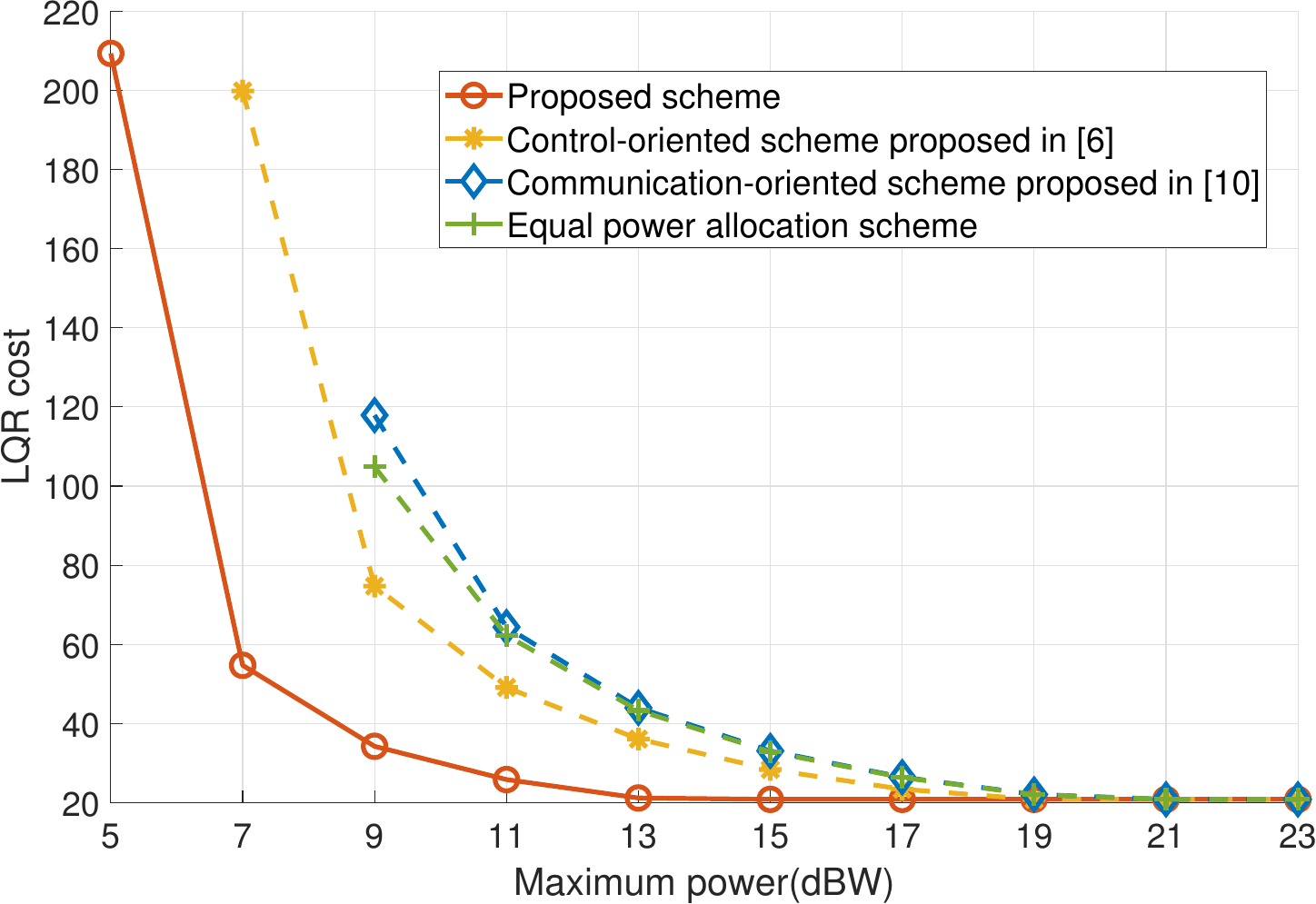}
		\caption{LQR cost achieved by different power allocation methods.}
		\label{fig:sim2}
	\end{figure}
	Fig. \ref{fig:sim2} compares the LQR cost achieved by the proposed scheme, the control-oriented scheme proposed in~\cite{Lei}, the communication-oriented scheme proposed in \cite{Hu} and the equal power allocation scheme. It shows that the proposed scheme achieves the lowest LQR cost for any given $P_{\text{max}}$. Particularly when $P_{\text{max}}<7$dBW, the conventional schemes lead to system instability, causing the cost to approach infinity. We observe that the control-oriented scheme proposed in \cite{Lei} achieves better performance in $\textbf{SC}^3$ closed loop control compared to the other two schemes. However, since it does not consider the OCE constraints, it suffers performance degradation relative to our scheme. The scheme in \cite{Hu}, which focuses solely on communication while neglecting closed loop and OCE, has a slightly higher LQR cost compared to the equal power allocation scheme, indicating that over-reliance on channel conditions may degrade control performance. Moreover, the LQR cost decreases as total power increases, demonstrating the benefits of enhancing downlink communication capability. However, as the total power becomes sufficiently large, the reduction in LQR cost gradually diminishes and the performance of different schemes converges. This is because all systems approach the performance upper bound imposed by OCE, beyond which merely increasing transmit power cannot overcome the bottleneck. This highlights the significance of the proposed OCE concept in practical scenarios and the superiority of the proposed scheme under resource-constrained conditions.
	
	\begin{figure} [t]
		\centering
		\includegraphics[width=1.0\linewidth]{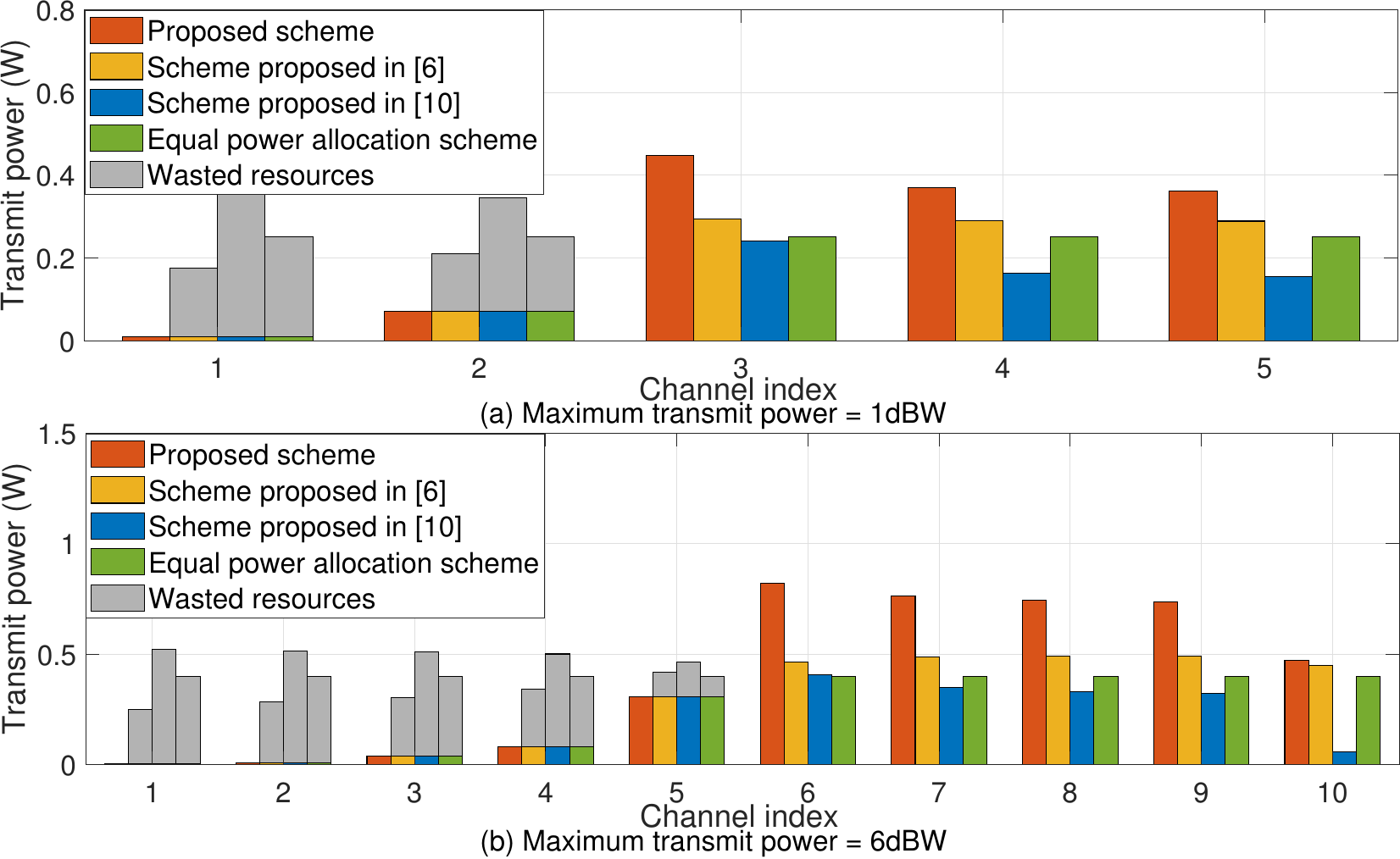}
		\caption{Power allocated to channels by different methods.}
		\label{fig:sim3}
	\end{figure}
	
	Fig. \ref{fig:sim3}(a) and \ref{fig:sim3}(b) compare the power allocated to each channel obtained with different methods. The channels are sorted such that $l_1\geq l_2 \geq \cdots \geq l_K$ and $E_1\leq E_2 \leq \cdots \leq E_K$. The results clearly illustrate the differences among the four methods. The proposed scheme balances OCE and channel gain, thereby maximizing resource utilization. Although the scheme in \cite{Lei} is control-oriented, it may allocate excessive power to LAAs with favorable channels but low OCE, while LAAs with stronger work capability but less favorable channels receive insufficient resources. The other two schemes similarly suffer from resource waste on channels with low OCE. In fact, based on (\ref{zuiyoujie1}), we can also know that the proposed scheme introduces OCE constraints to the water-filling framework, effectively avoiding inefficient allocation of communication resources.
	
	\begin{figure} [t]
		\centering
		\includegraphics[width=1.0\linewidth]{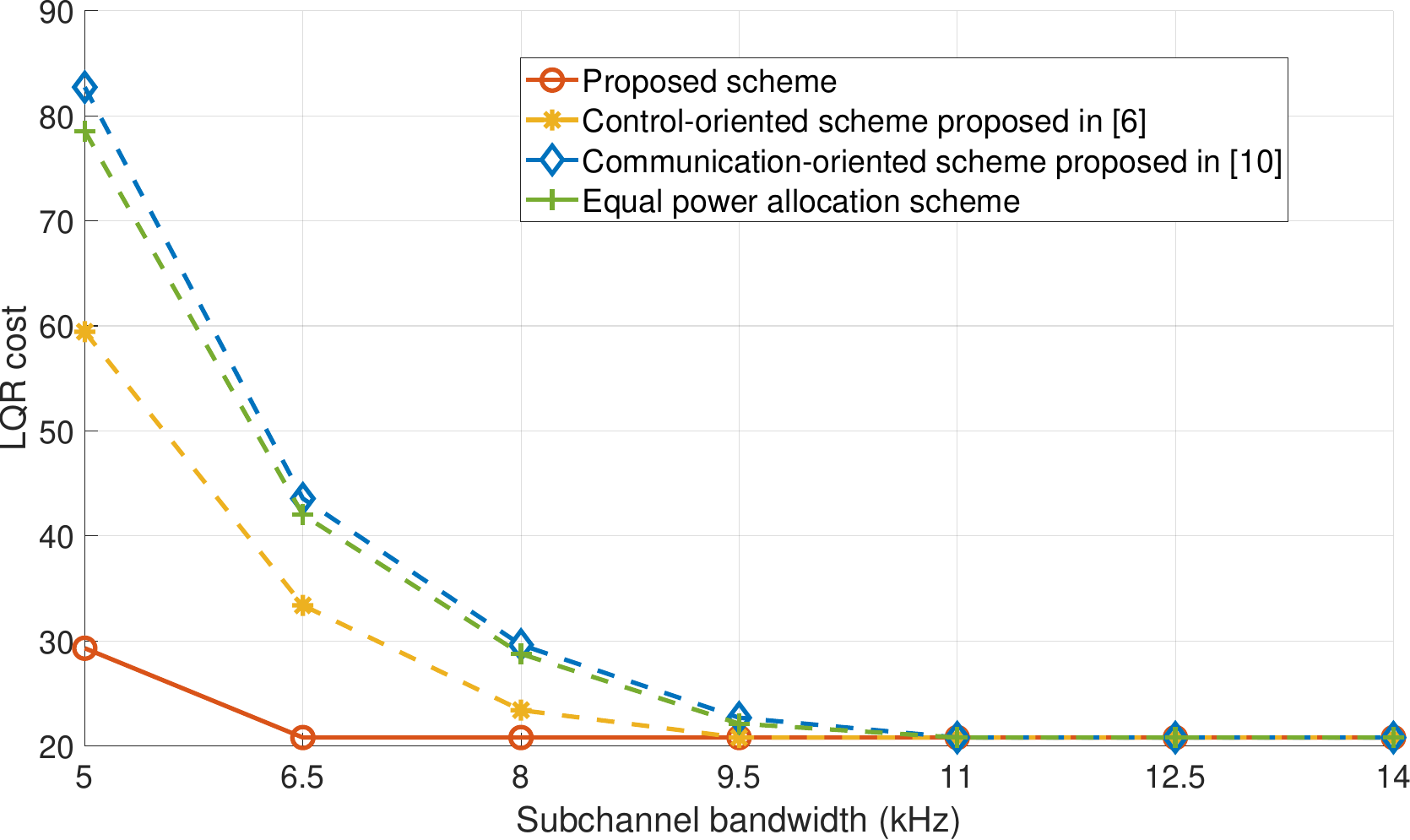}
		\caption{LQR cost under different subchannel bandwidths.}
		\label{fig:sim4}
	\end{figure}
	
	Fig. \ref{fig:sim4} investigates the impact of subchannel bandwidth $B_k$ on the LQR cost under $P_{\text{max}}=10$dBW. The results show that the proposed scheme achieves the lowest LQR cost for any given $B_k$, and its advantage becomes more pronounced when bandwidth is limited. As the $B_k$ increases, the LQR cost of all schemes decreases, indicating that, besides increasing transmit power, appropriately increasing the subchannel bandwidth is also an effective way to improve control performance. With the increase of $B_k$, the LQR cost of all schemes gradually converges, implying that the system bottleneck shifts from bandwidth to OCE. This is consistent with the convergence observed in Fig. \ref{fig:sim2}, jointly validating that OCE constitutes the ultimate upper bound on system performance. The proposed scheme is robust in both power and bandwidth dimensions. How to jointly allocate power and bandwidth to optimize $\textbf{SC}^3$ closed loop control performance remains an open and promising direction for future work.
	
	\section{Conclusions}
	In this letter, we have investigated a multi-LAA cooperative $\textbf{SC}^3$ closed loop control system in the low-altitude economy. We defined the OCE to quantify the work capability of LAAs during mission execution, and proposed a mission-oriented adaptive power allocation scheme that minimizes the LQR cost to improve mission efficiency. The original non-convex problem was decoupled into two convex subproblems, for which closed-form solutions were derived and an iterative algorithm was developed. Simulation results demonstrated that the proposed scheme effectively exploits limited resources by jointly accounting for OCE and channel heterogeneity, significantly outperforming traditional approaches and validating the value of using OCE in closed loop system design.


	%	\begin{thebibliography}{1}
		%		\bibliographystyle{IEEEtran}

		\newpage
		
		%If you have an EPS/PDF photo (graphicx package needed), extra braces are needed around the contents of the optional argument to biography to prevent the LaTeX parser from getting confused when it sees the complicated
		%$\backslash${\tt{includegraphics}} command within an optional argument. (You can create your own custom macro containing the $\backslash${\tt{includegraphics}} command to make things simpler here.)
		
		
		%\bf{If you include a photo:}\vspace{-33pt}
		%\begin{IEEEbiography}[{\includegraphics[width=1in,height=1.25in,clip,keepaspectratio]{fig1}}]{Michael Shell}
		%Use $\backslash${\tt{begin\{IEEEbiography\}}} and then for the 1st argument use $\backslash${\tt{includegraphics}} to declare and link the author photo.
		%Use the author name as the 3rd argument followed by the biography text.
		%\end{IEEEbiography}
		
		\vspace{11pt}
		

		\vfill
	\end{document}